\documentclass[11pt,a4paper]{article}
\usepackage[T1]{fontenc}
\usepackage{mathptmx}
\usepackage[utf8]{inputenc}
\usepackage[margin=2.5cm]{geometry}
\usepackage{amsmath,amssymb,amsthm,mathtools,bm}
\usepackage{graphicx}
\usepackage{booktabs}
\usepackage{authblk}
\usepackage[expansion=false]{microtype}
\usepackage{xcolor}
\usepackage[colorlinks=true,linkcolor=blue!50!black,citecolor=blue!50!black,urlcolor=blue!50!black]{hyperref}
\usepackage{cleveref}
\usepackage[english]{babel}

\theoremstyle{plain}
\newtheorem{proposition}{Proposition}
\newtheorem{lemma}[proposition]{Lemma}
\theoremstyle{remark}
\newtheorem{remark}{Remark}

\newcommand{\dd}{\mathrm{d}}
\newcommand{\rhoc}{\rho_{c}}
\newcommand{\astar}{a_{\ast}}
\newcommand{\Ric}{\mathcal{R}}
\newcommand{\Mzero}{\mathcal{M}_{0}}
\newcommand{\Xmax}{\Xi}

\title{Bouncing Dust Collapse and Black-to-White Hole Transition on the Brane}
\author[1]{Rikpratik Sengupta\thanks{rikpratik.sengupta@gmail.com}}
\author[2]{Chiranjeeb Singha\thanks{chiranjeeb.singha@iucaa.in}}
\affil[1]{Department of Physics, Indian Institute of Technology Kanpur, Kalyanpur 208016, India}
\affil[2]{Inter-University Centre for Astronomy and Astrophysics, Post Bag 4, Pune 411007, India}
\date{}

\begin{document}
\maketitle

\begin{abstract}
We study the spherical collapse of a homogeneous dust cloud on a three-brane in the Shtanov--Sahni braneworld, in which the extra dimension is timelike. In the high-energy regime the brane Friedmann equation acquires a negative quadratic term in the matter density, controlled by a critical density $\rhoc=2\lambda$, where $\lambda$ is the magnitude of the negative brane tension. For a spatially closed interior we integrate the dynamics exactly. Every turning point of the scale factor lies above $\astar=\sqrt{3/4\pi\lambda}$, a bounce exists if and only if a simple condition on the comoving density holds, and all curvature invariants stay below bounds fixed by $\lambda$ alone along the entire cycle. Comoving worldlines are geodesically complete. We then match the cloud across a comoving boundary to a spherically symmetric exterior with a trace-free effective source and solve the Israel conditions in closed form. The exterior is a charged Vaidya geometry whose mass function and tidal charge are fixed by the boundary radius. The tidal charge is positive, opposite in sign to the Randall--Sundrum case, and the exterior cannot be static. A trapped region forms on the boundary if and only if the mass exceeds $2/(3\sqrt{\pi\lambda})$, and in that case the bounce still takes place in an untrapped neighbourhood of the boundary. For trapped clouds we construct an exterior over the full contraction, bounce and re-expansion by joining an advanced and a retarded charged Vaidya patch through a static wedge. The matched metric is of class $C^{1,1}$, the effective energy conditions hold, and the cloud re-emerges through a white hole region into a second asymptotic region. In the general relativistic limit $\lambda\to\infty$ the Oppenheimer--Snyder solution is recovered.
\end{abstract}


\section{Introduction}

Gravitational collapse of a sufficiently massive object ends, in general relativity, in a spacetime singularity. The singularity theorems of Penrose and of Hawking and Penrose show that this outcome is generic once a trapped surface has formed and the energy conditions hold \cite{Penrose1965,HawkingPenrose1970}. The first explicit example is the Oppenheimer--Snyder (OS) model \cite{OS1939}, in which a homogeneous pressureless sphere described by a closed Friedmann--Lema\^itre--Robertson--Walker (FLRW) interior is matched to a Schwarzschild exterior. The cloud contracts to zero radius in finite proper time, a horizon forms on the way, and curvature diverges at the end. Inhomogeneous generalisations of the OS model, the Lema\^itre--Tolman--Bondi solutions \cite{Tolman1934,Bondi1947}, show a richer causal structure, including naked singularities for some initial data, but they share the same basic fate of the matter.

Since curvature becomes Planckian near the endpoint, the classical description cannot be trusted there, and a large literature asks whether quantum gravity corrections can replace the singularity by a bounce. Thin-shell analyses of bouncing collapse without commitment to a particular mechanism show that a bounce must occur at or outside the threshold of trapped-region formation if the exterior is Schwarzschild and the junction is minimal \cite{BenAchour2020}. Effective models inspired by loop quantum cosmology produce pulsating compact objects with a Planckian curvature bound \cite{BenAchourUzan2020}. In these examples the repulsive effect that halts the collapse is of quantum origin.

A different route to bounded curvature uses classical higher-dimensional geometry. In the Randall--Sundrum (RS) braneworld \cite{RS1,RS2} our universe is a three-brane in a five-dimensional anti-de Sitter bulk, and the effective four-dimensional equations contain a term quadratic in the stress-energy and a projection of the bulk Weyl tensor \cite{SMS2000,Maartens2004}. Dadhich, Maartens, Papadopoulos and Rezania obtained a static black hole on such a brane whose metric has the Reissner--Nordstr\"om form with a tidal charge instead of an electric one \cite{Dadhich2000}. For collapse, Bruni, Germani and Maartens showed that a homogeneous dust interior cannot be matched to a static exterior \cite{Bruni2001}, and Govender and Dadhich argued that the exterior must then carry radiation of Vaidya type \cite{Govender2002}. Related work studied braneworld stars and collapse with a radiative bulk \cite{Germani2001,Pal2006,Wiseman2002,Deruelle2001,Visser2003,Creek2006}. In all of these RS-type settings the quadratic correction is attractive, so it strengthens rather than halts the collapse.

Shtanov and Sahni proposed a dual construction in which the extra dimension is timelike and the bulk has a positive cosmological constant \cite{ShtanovSahni2003}. The brane then carries negative tension, and the sign of the quadratic correction is reversed. The resulting cosmology has a non-singular bounce for matter satisfying standard energy conditions, and the bounce can recur, which provides cyclic evolution \cite{Kanekar2001,SahniToporensky2012,Sahni2015}. The same framework supports traversable wormholes with finite tidal forces \cite{Sengupta2023wh}. An obvious question is whether the sign reversal also changes the fate of a collapsing cloud. This is the question we address.

We consider the OS configuration on the Shtanov--Sahni (SS) brane: a closed FLRW dust interior glued along a comoving timelike surface to a spherically symmetric exterior. The main results are the following.
\begin{enumerate}
\item The interior dynamics is integrated exactly. At every turning point the scale factor is bounded below by $\astar=\sqrt{3/4\pi\lambda}$, which depends only on the brane tension. A bounce exists if and only if the comoving density exceeds an explicit threshold. For a given cloud the actual bounce radius is determined by its comoving density and generally lies above $\astar$ (\cref{sec:interior}).
\item The Ricci scalar, the squared Ricci tensor and the Kretschmann scalar of the interior are bounded along the whole trajectory by constants proportional to powers of $\lambda$, and comoving worldlines are geodesically complete (\cref{sec:curv}).
\item If the exterior is of Vaidya form with a trace-free effective source, the Israel junction conditions have a unique solution, which we give in closed form. The exterior is a charged Vaidya geometry with a positive tidal charge, and it cannot be static (\cref{sec:exterior}).
\item Trapped surfaces form on the boundary if and only if the cloud mass exceeds $2/(3\sqrt{\pi\lambda})$, independently of the cloud radius. When they do form, the bounce occurs while the boundary is untrapped (\cref{sec:trapping}).
\item For trapped clouds the boundary lies between the apparent horizons during the trapped interval and inside the inner one at the bounce. An advanced exterior for the contraction and a retarded exterior for the expansion can be joined through a static wedge into a single exterior with continuous curvature, bounded curvature scalars and non-negative effective energy flux. Distant observers see a black hole, and the cloud re-expands into a second asymptotic region (\cref{sec:global}).
\end{enumerate}
Section~\ref{sec:num} gives a numerical illustration, and \cref{sec:disc} discusses the results and their scope.

\paragraph{Conventions.} We use signature $(-,+,+,+)$ on the brane, $(-,-,+,+,+)$ in the bulk, and units $G_N=c=1$. Greek indices label brane coordinates. The Ricci scalar is written $\Ric$ to distinguish it from the areal radius $R$.

\section{Effective brane equations}\label{sec:framework}

\subsection{Field equations}

The SS model consists of a four-dimensional brane $\mathcal B$ in a five-dimensional bulk $\mathcal M$ of signature $(-,-,+,+,+)$, with action \cite{ShtanovSahni2003}
\begin{equation}
S=M^{3}\!\int_{\mathcal M}\!\sqrt{-g}\,(\mathcal R_5-2\Lambda_5)\,\dd^{5}x+2M^3\!\int_{\mathcal B}\!\sqrt{-h}\,K\,\dd^{4}x+\int_{\mathcal B}\!\sqrt{-h}\,\bigl(m^{2}\Ric-2\sigma+L_{\rm m}\bigr)\dd^{4}x ,
\label{eq:action}
\end{equation}
where $\Lambda_5>0$ is the bulk cosmological constant, $\sigma<0$ the brane tension, $M$ and $m$ the five- and four-dimensional Planck masses, $K$ the trace of the extrinsic curvature and $L_{\rm m}$ the matter Lagrangian. The unit normal $n^A$ to the brane is timelike, $n^An_A=-1$. We write
\begin{equation}
\lambda\equiv|\sigma|=-\sigma>0 ,\qquad \rhoc\equiv 2\lambda .
\label{eq:lambda}
\end{equation}

Projecting the bulk equations on the brane with the Gauss--Codazzi relations and imposing $\mathbb Z_2$ symmetry yields effective four-dimensional equations \cite{SMS2000,Maartens2004,ShtanovSahni2003,Sahni2005}. We use them in the high-energy regime in which the induced-gravity term $m^2\Ric$ is subdominant and the effective four-dimensional cosmological constant has been tuned to zero. This is the regime in which the bounce of the SS cosmology occurs. In this regime the equations read
\begin{equation}
G_{\mu\nu}=8\pi T_{\mu\nu}-\frac{48\pi}{\lambda}\,\pi_{\mu\nu}-\mathcal E_{\mu\nu},
\label{eq:EFE}
\end{equation}
with
\begin{equation}
\pi_{\mu\nu}=-\tfrac14 T_{\mu\alpha}T^{\alpha}{}_{\nu}+\tfrac1{12}T\,T_{\mu\nu}+\tfrac18 h_{\mu\nu}T_{\alpha\beta}T^{\alpha\beta}-\tfrac1{24}h_{\mu\nu}T^{2}.
\label{eq:pi}
\end{equation}
The tensor $\mathcal E_{\mu\nu}$ is the projection of the bulk Weyl tensor on the brane, with the sign convention chosen so that it enters \eqref{eq:EFE} as shown. It is symmetric and trace-free,
\begin{equation}
\mathcal E_{\mu\nu}=\mathcal E_{\nu\mu},\qquad \mathcal E^{\mu}{}_{\mu}=0 ,
\label{eq:Etraceless}
\end{equation}
and the contracted Bianchi identity together with $D^\mu T_{\mu\nu}=0$ gives
\begin{equation}
D^{\mu}\mathcal E_{\mu\nu}=-\frac{48\pi}{\lambda}D^{\mu}\pi_{\mu\nu}.
\label{eq:Ediv}
\end{equation}
Relative to RS II, the quadratic term enters with the opposite sign because the normal is timelike. For $\lambda\to\infty$ the equations reduce to general relativity with a Weyl term that is no longer constrained by the matter. The factor in front of $\pi_{\mu\nu}$ is fixed by requiring that a homogeneous isotropic dust universe obeys the Friedmann equation \eqref{eq:Friedmann} below with critical density $\rhoc=2\lambda$.

In vacuum $T_{\mu\nu}=0$, hence $\pi_{\mu\nu}=0$, and \eqref{eq:EFE} becomes
\begin{equation}
G_{\mu\nu}=-\mathcal E_{\mu\nu},\qquad \Ric=0 ,\qquad D^\mu\mathcal E_{\mu\nu}=0 .
\label{eq:vacuum}
\end{equation}
The brane vacuum is therefore not the four-dimensional Einstein vacuum: the geometry is sourced by the trace-free, conserved tensor $\mathcal E_{\mu\nu}$, which encodes the non-local bulk influence. The only restrictions that the brane equations place on $\mathcal E_{\mu\nu}$ are \eqref{eq:Etraceless} and \eqref{eq:Ediv}. Whether a given $\mathcal E_{\mu\nu}$ is realised by a regular bulk solution is a separate question that is not addressed in this paper.

The effective equations are used for all densities up to $\rhoc$. The induced-gravity term and higher-order corrections in the extrinsic curvature are neglected, so the results below hold to the extent that these corrections remain small as the density approaches $\rhoc$, which is the regime in which the bounce of the SS cosmology is described \cite{ShtanovSahni2003}. The bulk is taken to be $\mathbb Z_2$ symmetric about the brane.

\subsection{The timelike extra dimension}\label{sec:timelike}

The sign of the quadratic term is the origin of the results below, so we state what it rests on. In RS II the bulk is anti-de Sitter, the extra dimension is spacelike and the tension is positive. The SS model is obtained by changing the signature of the extra dimension and the sign of the bulk cosmological constant. The Gauss equation then carries the sign $n^An_A=-1$ in front of the quadratic extrinsic-curvature terms, and the Israel condition across the brane relates the extrinsic curvature to the matter content with the corresponding sign. The net effect is that the correction $\propto\rho^2$ to the Friedmann equation is negative, as in \eqref{eq:Friedmann}, and that the tension must be negative if the effective gravitational coupling is to be positive \cite{ShtanovSahni2003,Sahni2005}.

A timelike extra dimension is known to lead to tachyonic Kaluza--Klein gravitons in the original SS proposal \cite{ShtanovSahni2003}. Iglesias and Kakushadze have shown that this pathology is not universal: in a five-dimensional Einstein--Hilbert--Gauss--Bonnet setting with a second timelike direction there is a solitonic three-brane that localises gravity completely, and no tachyonic or negative-norm states appear \cite{Iglesias2001}. We take the effective equations \eqref{eq:EFE} as the working description of the brane and do not analyse the Kaluza--Klein spectrum of the specific background used here. This is a limitation of the present work, which we return to in \cref{sec:disc}.

\section{The interior: a closed dust cloud}\label{sec:interior}

\subsection{Dynamics}

The interior metric is the closed FLRW form
\begin{equation}
ds^{2}=-\dd\tau^{2}+a(\tau)^{2}\bigl(\dd\chi^{2}+\sin^{2}\chi\,\dd\Omega^{2}\bigr),
\label{eq:interior}
\end{equation}
filled with dust, $T_{\mu\nu}=\rho\,u_\mu u_\nu$, so that $\dot\rho+3(\dot a/a)\rho=0$ and
\begin{equation}
\rho=\frac{\rho_0}{a^{3}},\qquad \rho_0=\text{const}>0 .
\end{equation}
For dust $\pi_{00}=\rho^2/12$ and $\pi_{ij}=(\rho^2/12)\,h_{ij}$, so, with the Weyl projection taken to vanish inside the cloud, the $00$ component of \eqref{eq:EFE} is
\begin{equation}
\Bigl(\frac{\dot a}{a}\Bigr)^{2}+\frac1{a^{2}}=\frac{8\pi}{3}\,\rho\Bigl(1-\frac{\rho}{\rhoc}\Bigr),\qquad \rhoc=2\lambda ,
\label{eq:Friedmann}
\end{equation}
and the second Friedmann equation follows from \eqref{eq:Friedmann} and the conservation law,
\begin{equation}
\frac{\ddot a}{a}=-\frac{4\pi}{3}\,\rho\Bigl(1-\frac{4\rho}{\rhoc}\Bigr).
\label{eq:Ray}
\end{equation}
The vanishing of the interior Weyl projection removes the dark radiation term $C/a^4$ that would otherwise appear on the right side of \eqref{eq:Friedmann}. A non-zero $C$ changes the first integral \eqref{eq:potential} below, and with it the turning points and the junction solution, so the results apply to $C=0$. The brane equations are therefore equivalent to Einstein's equations with the effective perfect fluid
\begin{equation}
\rho_{\rm eff}=\rho\Bigl(1-\frac{\rho}{\rhoc}\Bigr),\qquad p_{\rm eff}=-\frac{\rho^{2}}{\rhoc}.
\label{eq:effective}
\end{equation}
The effective pressure is negative and grows with density. The effective strong energy condition $\rho_{\rm eff}+3p_{\rm eff}=\rho(1-4\rho/\rhoc)\ge0$ fails for $\rho>\rhoc/4$, and the effective null energy condition $\rho_{\rm eff}+p_{\rm eff}=\rho(1-2\rho/\rhoc)\ge0$ fails for $\rho>\rhoc/2=\lambda$.

Define the constant $\Mzero\equiv 8\pi\rho_0/3$ (a length in our units). Equation \eqref{eq:Friedmann} becomes a one-dimensional problem,
\begin{equation}
\dot a^{2}=g(a)-1,\qquad g(a)\equiv\frac{\Mzero}{a}-\frac{\Mzero\,\rho_0}{\rhoc\,a^{4}},\qquad \ddot a=\tfrac12 g'(a).
\label{eq:potential}
\end{equation}
The last relation follows by differentiation and agrees with \eqref{eq:Ray}. In the limit $\rhoc\to\infty$, $g=\Mzero/a$ and \eqref{eq:potential} is the OS cycloid with maximal radius $\Mzero$ and collapse to $a=0$ in the proper time $\pi\Mzero/2$.

We need three characteristic scales,
\begin{equation}
a_z=\Bigl(\frac{\rho_0}{\rhoc}\Bigr)^{1/3},\qquad a_m=\Bigl(\frac{4\rho_0}{\rhoc}\Bigr)^{1/3},\qquad \Xmax\equiv g(a_m)=2\pi\Bigl(\frac{\rho_0^{2}\lambda}{2}\Bigr)^{1/3}.
\label{eq:scales}
\end{equation}
The function $g$ vanishes at $a_z$, increases up to its maximum $\Xmax$ at $a_m=2^{2/3}a_z$, where $\rho=\rhoc/4$, and decreases to zero as $a\to\infty$.

\subsection{Turning points and the universal bound}

\begin{proposition}[Density bound and universal lower bound]\label{prop:bound}
Let $a(\tau)$ solve \eqref{eq:Friedmann}. Then
\begin{enumerate}
\item[(i)] $0<\rho<\rhoc$ at all times;
\item[(ii)] at any turning point $\dot a=0$ one has $a\ge\astar\equiv\sqrt{3/(4\pi\lambda)}$, with equality if and only if $\rho=\lambda$ at that point;
\item[(iii)] at a turning point $\ddot a>0$ if and only if $\rho>\rhoc/4$.
\end{enumerate}
\end{proposition}
\begin{proof}
(i) The left side of \eqref{eq:Friedmann} is at least $1/a^2>0$, hence $\rho(1-\rho/\rhoc)>0$. (ii) At a turning point $1/a^2=(8\pi/3)\rho(1-\rho/\rhoc)$. The right side is maximal at $\rho=\rhoc/2=\lambda$, where it equals $2\pi\rhoc/3=4\pi\lambda/3$. Hence $1/a^2\le4\pi\lambda/3$. (iii) follows from \eqref{eq:Ray}.
\end{proof}

The bound $\astar$ is universal: it depends on $\lambda$ only. It is a lower bound on the radius at which a turning point can occur, not the bounce radius of a given cloud. The actual radius is a root of $g(a)=1$ and depends on $\rho_0$, as the next proposition shows.

\begin{proposition}[Closed dust cloud]\label{prop:cycle}
Let $\Xmax$ be defined by \eqref{eq:scales}.
\begin{enumerate}
\item[(a)] If $\Xmax<1$ there is no solution of \eqref{eq:Friedmann} with $\dot a^2\ge0$. Equivalently, solutions exist only for $\rho_0\ge\rho_{0,\min}\equiv 1/(2\pi^{3/2}\sqrt\lambda)$.
\item[(b)] If $\Xmax=1$ the only solution is the static one, $a=a_m$, $\rho=\rhoc/4$.
\item[(c)] If $\Xmax>1$ the equation $g(a)=1$ has exactly two roots $a_-<a_+$ with $a_z<a_-<a_m<a_+$. Both are simple, $\ddot a(a_-)>0$ and $\ddot a(a_+)<0$, and the scale factor oscillates periodically between $a_-$ and $a_+$ with period
\begin{equation}
T=2\int_{a_-}^{a_+}\frac{\dd a}{\sqrt{g(a)-1}}<\infty .
\label{eq:period}
\end{equation}
The bounce radius satisfies $a_-\ge\astar$ with equality only for $\rho_0=\lambda\astar^{3}$, and the bounce density $\rho_b=\rho_0/a_-^3$ lies in $(\rhoc/4,\rhoc)$ and equals $\lambda$ only in that case.
\end{enumerate}
\end{proposition}
\begin{proof}
From the shape of $g$ stated above, $\dot a^2=g-1\ge0$ requires $g(a)\ge1$, which is possible if and only if $\Xmax\ge1$. The condition $\Xmax\ge1$ is equivalent to $\rho_0^2\ge 1/(4\pi^3\lambda)$, which gives (a). For $\Xmax=1$ the set $\{g\ge1\}$ is the single point $a_m$, and there $g'(a_m)=0$, so $\ddot a=0$. For $\Xmax>1$ the monotonicity of $g$ on either side of $a_m$ gives two simple roots, with $a_->a_z$ because $g(a_z)=0<1$. Since $\ddot a=g'/2$, $\ddot a(a_-)>0>\ddot a(a_+)$. The one-dimensional problem $\ddot a=g'(a)/2$ has a smooth global solution confined to $[a_-,a_+]$, which reflects at the simple roots in finite time, so the motion is periodic, and the integral \eqref{eq:period} converges because the roots are simple. By \Cref{prop:bound}(ii), $a_-\ge\astar$. If $\rho_0=\lambda\astar^3$ then $g(\astar)=\astar^{2}(8\pi/3)\lambda(1-\tfrac12)=1$, and $\astar<a_m=2^{1/3}\astar$ identifies $a_-=\astar$. Finally $a_-<a_m$ gives $\rho_b>\rho_0/a_m^3=\rhoc/4$, and \Cref{prop:bound}(i) gives $\rho_b<\rhoc$.
\end{proof}

For large clouds, $\rho_0\gg\rhoc$, expanding $g(a)=1$ near $a_z$ gives
\begin{equation}
a_-=a_z\Bigl[1+\frac{a_z}{3\Mzero}+O\bigl((a_z/\Mzero)^2\bigr)\Bigr],\qquad \rho_b=\rhoc\Bigl[1-\frac{a_z}{\Mzero}+\dots\Bigr].
\label{eq:largecloud}
\end{equation}
The bounce density approaches $\rhoc$ but never reaches it. The fixed value $\rho=\lambda$ obtained for the saturating cloud $\rho_0=\lambda\astar^3$ is a property of that particular configuration. It is not a universal critical density.

Figure~\ref{fig:one} shows $\dot a^2$ as a function of $a$ for $\rho_0=\lambda$ together with the bound $\astar$, and the dependence of $a_-$ and $\rho_b/\rhoc$ on $\rho_0$.

\begin{figure}[t]
\centering
\includegraphics[width=0.96\textwidth]{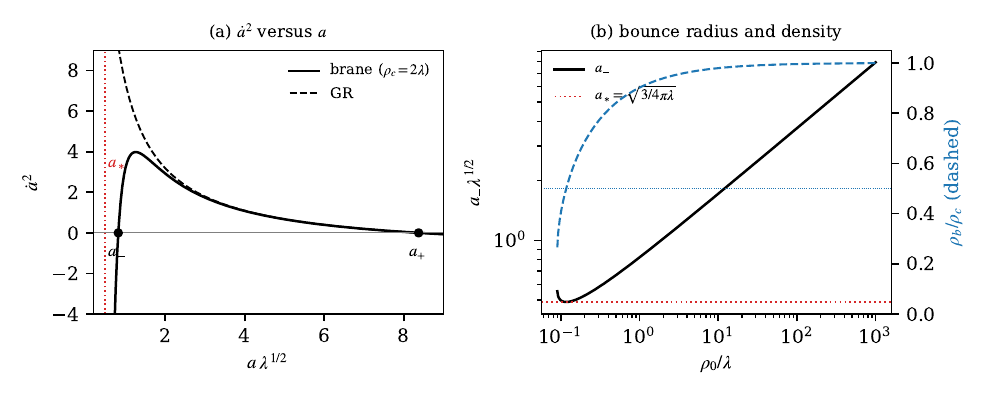}
\caption{(a) The function $\dot a^{2}=g(a)-1$ for $\rho_0=\lambda$ in units $\lambda=1$ (solid) and its general relativistic limit $\Mzero/a-1$ (dashed). The turning points $a_\mp$ and the universal bound $\astar$ (dotted) are marked. (b) Bounce radius $a_-$ (solid, left axis) and bounce density $\rho_b/\rhoc$ (dashed, right axis) as functions of $\rho_0$. The radius never goes below $\astar$ and the density never reaches $\rhoc$.}
\label{fig:one}
\end{figure}

\subsection{Curvature along the whole trajectory}\label{sec:curv}

Write $X\equiv\ddot a/a$ and $Y\equiv\dot a^2/a^2+1/a^2$. For the metric \eqref{eq:interior} the standard kinematic identities are
\begin{equation}
\Ric=6(X+Y),\qquad R_{\mu\nu}R^{\mu\nu}=12\,(X^{2}+XY+Y^{2}),\qquad R_{\alpha\beta\gamma\delta}R^{\alpha\beta\gamma\delta}=12\,(X^{2}+Y^{2}).
\label{eq:invariants}
\end{equation}
Using \eqref{eq:Friedmann} and \eqref{eq:Ray},
\begin{equation}
X=-\frac{4\pi}{3}\rho\Bigl(1-\frac{4\rho}{\rhoc}\Bigr),\qquad Y=\frac{8\pi}{3}\rho\Bigl(1-\frac{\rho}{\rhoc}\Bigr),
\label{eq:XY}
\end{equation}
so that all three invariants are functions of the density alone.

\begin{proposition}[Bounded curvature]\label{prop:curv}
Along any solution, $-\pi\rhoc/12\le X<4\pi\rhoc$ and $0<Y\le 2\pi\rhoc/3$. Consequently
\begin{equation}
|\Ric|<28\pi\rhoc=56\pi\lambda,\qquad R_{\mu\nu}R^{\mu\nu}<\tfrac{688}{3}\pi^2\rhoc^{2},\qquad R_{\alpha\beta\gamma\delta}R^{\alpha\beta\gamma\delta}<\tfrac{592}{3}\pi^{2}\rhoc^{2}.
\end{equation}
At the bounce of a cloud with bounce density $\rho_b$,
\begin{equation}
\Ric_b=8\pi\rho_b\Bigl(1+\frac{2\rho_b}{\rhoc}\Bigr),
\end{equation}
and for the saturating cloud $\Ric_b=16\pi\lambda$, $R_{\mu\nu}R^{\mu\nu}=64\pi^{2}\lambda^{2}$, $R_{\alpha\beta\gamma\delta}R^{\alpha\beta\gamma\delta}=\tfrac{128}{3}\pi^2\lambda^2$.
\end{proposition}
\begin{proof}
By \Cref{prop:bound}(i), $0<\rho<\rhoc$. The polynomial $\rho(1-4\rho/\rhoc)$ takes values in $(-3\rhoc,\rhoc/16]$ on this interval, which gives the range of $X$. The polynomial $\rho(1-\rho/\rhoc)$ takes values in $(0,\rhoc/4]$, which gives the range of $Y$. The bounds on the invariants follow from \eqref{eq:invariants} and $|X|<4\pi\rhoc$, $0<Y\le2\pi\rhoc/3$. At the bounce $\dot a=0$, so $Y_b=1/a_-^2$, and $\Ric_b=6(X_b+Y_b)=8\pi\rho_b(1+2\rho_b/\rhoc)$ by \eqref{eq:XY}. For the saturating cloud, $\rho_b=\lambda$ and $X_b=Y_b=4\pi\lambda/3$.
\end{proof}

Every invariant is therefore bounded by a constant of order $\lambda$, $\lambda^2$ along the entire cycle, not only at the turning point, and each diverges only in the limit $\lambda\to\infty$, where \eqref{eq:XY} reduces to the OS behaviour $X,Y\propto\rho\to\infty$.

\begin{proposition}[Completeness of comoving worldlines]\label{prop:complete}
For $\Xmax>1$ the solution $a(\tau)$ extends to a smooth periodic function on $\tau\in(-\infty,\infty)$ with $a_-\le a(\tau)\le a_+$. The comoving worldlines $\chi,\theta,\phi=\text{const}$ are timelike geodesics with proper time $\tau$ and are complete in both time directions.
\end{proposition}
\begin{proof}
The equation $\ddot a=g'(a)/2$ has a smooth right side on $a>0$, and the orbit stays in the compact set $[a_-,a_+]\subset(0,\infty)$, so the solution exists for all $\tau$. In the metric \eqref{eq:interior}, $\Gamma^{\mu}{}_{\tau\tau}=0$, so the curves are geodesics parametrised by $\tau$.
\end{proof}
This statement concerns the radial comoving congruence. Completeness for non-comoving geodesics is not shown.

Figure~\ref{fig:two} shows the scale factor and the three invariants for $\rho_0=\lambda$.

\begin{figure}[t]
\centering
\includegraphics[width=0.96\textwidth]{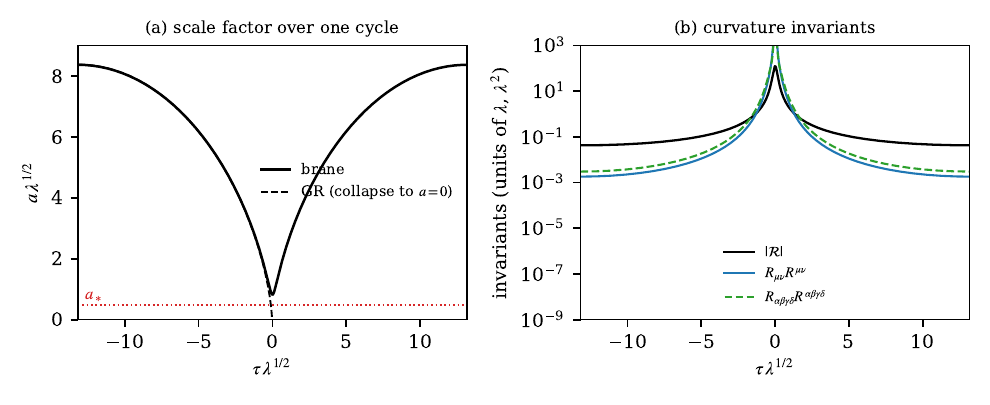}
\caption{(a) Scale factor over one cycle for $\rho_0=\lambda$ (solid, units $\lambda=1$), compared with the general relativistic collapse from the same maximal radius (dashed), which reaches $a=0$. The two curves agree until the density becomes comparable with $\rhoc$. (b) The three curvature invariants along the cycle, with the maximum at the bounce.}
\label{fig:two}
\end{figure}

\subsection{Mass of the cloud}

The comoving boundary of the cloud is $\chi=\chi_\Sigma$ with $0<\chi_\Sigma<\pi/2$, which means that the cloud occupies less than a hemisphere of the closed three-space. This condition is needed for the exterior areal radius to increase outward. The areal radius of the boundary and the conserved mass are
\begin{equation}
R_\Sigma(\tau)=a(\tau)\sin\chi_\Sigma,\qquad M\equiv\frac{4\pi}{3}\rho_0\sin^{3}\chi_\Sigma .
\label{eq:Mdef}
\end{equation}
$M$ is the Oppenheimer--Snyder mass, $2M=(8\pi/3)\rho R_\Sigma^{3}$, and it is conserved because $\rho R_\Sigma^3=\rho_0\sin^3\chi_\Sigma$. We also define
\begin{equation}
\kappa\equiv\frac{8\pi}{3}\frac{\rho_0^{2}\sin^{6}\chi_\Sigma}{\rhoc}=\frac{3M^{2}}{4\pi\lambda}.
\label{eq:kappa}
\end{equation}

\section{The exterior and the junction conditions}\label{sec:exterior}

\subsection{Trace-free exteriors are charged Vaidya}

We take the exterior to be spherically symmetric and of Bondi form with $g_{uR}=-1$, the class used in earlier studies of radiating collapse on the brane \cite{Govender2002}. It contains the static charged solution as well as the retarded and advanced Vaidya geometries. Exteriors with a further metric function multiplying $\dd u\,\dd R$ are not considered. In the retarded form,
\begin{equation}
ds^{2}=-f(u,R)\,\dd u^{2}-2\,\dd u\,\dd R+R^{2}\dd\Omega^{2},\qquad f=1-\frac{2m(u,R)}{R},
\label{eq:Bondi}
\end{equation}
a direct computation gives (Appendix~\ref{app:A})
\begin{equation}
\Ric=\frac{2}{R^{3}}\,\partial_R\!\bigl(R^{2}\partial_R m\bigr).
\end{equation}
By \eqref{eq:vacuum} the exterior has $\Ric=0$, so $R^2\partial_Rm$ is a function of $u$ only, $\partial_R m=c(u)/R^{2}$ for some $c(u)$, and $m=m(u)-q(u)/(2R)$ with $q=2c$. Therefore
\begin{equation}
f(u,R)=1-\frac{2m(u)}{R}+\frac{q(u)}{R^{2}},
\label{eq:chargedVaidya}
\end{equation}
the charged Vaidya metric. Its Einstein tensor is
\begin{equation}
G_{\mu\nu}=\frac{q}{R^{4}}\bigl(-h^{(2)}_{\mu\nu}+\Omega_{\mu\nu}\bigr)+\Theta\,k_\mu k_\nu,\qquad \Theta\equiv\frac{q_u-2R\,m_u}{R^{3}},\quad k_\mu=\partial_\mu u,
\label{eq:Gvaidya}
\end{equation}
where $h^{(2)}=-f\dd u^2-2\dd u\dd R$, $\Omega=R^2\dd\Omega^2$, and subscript $u$ denotes $\partial_u$. The first term is of Maxwell type with charge parameter $q$. The second is a null flux with $k^\mu k_\mu=0$. Since $\Ric=0$ identically, \eqref{eq:vacuum} is satisfied with $\mathcal E_{\mu\nu}=-G_{\mu\nu}$, which is trace-free and, by the Bianchi identity, conserved. The Kretschmann scalar is $R_{\alpha\beta\gamma\delta}R^{\alpha\beta\gamma\delta}=48m^2/R^6-96mq/R^7+56q^2/R^8$.

\subsection{Junction conditions}

Let $\Sigma$ be the timelike hypersurface $\chi=\chi_\Sigma$ of the interior, with induced metric $-\dd\tau^2+R_\Sigma(\tau)^2\dd\Omega^2$. From the exterior it is the curve $(u(\tau),R_\Sigma(\tau))$, with unit tangent $U^\mu=(\dot u,\dot R)$, where an overdot is $\dd/\dd\tau$. Continuity of the induced metric requires
\begin{equation}
f\,\dot u^{2}+2\dot u\dot R=1 ,
\label{eq:first}
\end{equation}
evaluated at $R=R_\Sigma$. The outward unit normal is $n_\mu=(-\dot R,\dot u)$, with $n^R=1/\dot u-\dot R$.

We assume that $\Sigma$ carries no effective surface layer \cite{Israel1966,Bruni2001}, so that the extrinsic curvature of $\Sigma$ must be continuous. For a spherically symmetric surface this has three components. In the interior, $K_{\tau\tau}=0$ and $K^\theta{}_\theta=\cot\chi_\Sigma/a$. In the exterior, $K^\theta{}_\theta=n^R/R_\Sigma$. The Codazzi equation links $K^\tau{}_\tau$ to $\partial_\tau K^\theta{}_\theta$, to $(\dot R/R)(K^\tau{}_\tau-K^\theta{}_\theta)$ and to the momentum flux $G_{\mu\nu}n^\mu U^\nu$. Hence, wherever $\dot R\ne0$, continuity of $K^\tau{}_\tau$ is equivalent to continuity of the momentum flux, which vanishes in the interior because the interior effective fluid is isotropic and at rest in the comoving frame. At the isolated instants with $\dot R=0$ continuity follows from continuity in time. The conditions are therefore:

\begin{enumerate}
\item[(J1)] $f\dot u^2+2\dot u\dot R=1$,
\item[(J2)] $n^R=\cos\chi_\Sigma$ on $\Sigma$, from $K^\theta{}_\theta$ with $R_\Sigma=a\sin\chi_\Sigma$,
\item[(J3)] $G_{\mu\nu}n^\mu U^\nu=0$ on $\Sigma$.
\end{enumerate}

Equation (J2) and (J1) give
\begin{equation}
\dot u=\frac{1}{\cos\chi_\Sigma+\dot R_\Sigma},\qquad f\big|_\Sigma\equiv f_\Sigma=\cos^{2}\chi_\Sigma-\dot R_\Sigma^{2}=1-\sin^{2}\chi_\Sigma\,(1+\dot a^{2}).
\label{eq:udot}
\end{equation}
Using \eqref{eq:Friedmann}, $\sin^2\chi_\Sigma(1+\dot a^2)=(8\pi/3)\rho_{\rm eff}R_\Sigma^2$, so that
\begin{equation}
f_\Sigma=1-\frac{8\pi}{3}\rho_{\rm eff}R_\Sigma^{2}\qquad\Longleftrightarrow\qquad 2m-\frac{q}{R_\Sigma}=\frac{8\pi}{3}\rho_{\rm eff}R_\Sigma^{3}=2M-\frac{\kappa}{R_\Sigma^{3}},
\label{eq:A}
\end{equation}
where we used $\rho R_\Sigma^3=\rho_0\sin^3\chi_\Sigma$ and \eqref{eq:Mdef}, \eqref{eq:kappa}. The left side of the second equation in \eqref{eq:A} is twice the Misner--Sharp mass of the exterior evaluated on $\Sigma$. For the charged Vaidya metric, (J3) evaluates to
\begin{equation}
G_{\mu\nu}n^{\mu}U^{\nu}=\frac{\dot u^{2}}{R^{3}}\bigl(2R\,m_u-q_u\bigr)=0 .
\label{eq:J3}
\end{equation}
Since $\dot u>0$ in the chart, this is $\Theta|_\Sigma=0$: the null flux must vanish on the boundary.

\subsection{Solution}

\begin{proposition}[Exterior of the cloud]\label{prop:exterior}
Let $0<\chi_\Sigma<\pi/2$ and consider an interval of $\tau$ on which $\dot a>-\cot\chi_\Sigma$ (so that $\dot u>0$ in \eqref{eq:udot}). Then the junction conditions (J1) to (J3) with an exterior of the form \eqref{eq:chargedVaidya} have a unique solution,
\begin{equation}
m(u)=M+\frac{\kappa}{R_\Sigma^{3}}=M\Bigl(1+\frac{2\rho}{\rhoc}\Bigr),\qquad q(u)=\frac{3\kappa}{R_\Sigma^{2}}=\frac{9M^{2}}{4\pi\lambda R_\Sigma^{2}}=\frac{6M\rho R_\Sigma}{\rhoc},
\label{eq:solution}
\end{equation}
where $R_\Sigma=R_\Sigma(u)$ and $\rho=\rho_0/a^3$ is the density of the cloud at the retarded time $u$. The null flux off the boundary is
\begin{equation}
\Theta(u,R)=\frac{6\kappa}{R^{3}R_\Sigma^{4}}\frac{\dd R_\Sigma}{\dd u}\,\bigl(R-R_\Sigma\bigr),\qquad \frac{\dd R_\Sigma}{\dd u}=\dot R_\Sigma\bigl(\cos\chi_\Sigma+\dot R_\Sigma\bigr).
\label{eq:Theta}
\end{equation}
Moreover $G_{\mu\nu}n^\mu n^\nu=8\pi p_{\rm eff}=-8\pi\rho^2/\rhoc$ holds on $\Sigma$, in agreement with the interior.
\end{proposition}
\begin{proof}
Solve \eqref{eq:A} for $q=2(m-M)R_\Sigma+\kappa/R_\Sigma^2$ and insert it into $\dot q-2R_\Sigma\dot m=0$, which is \eqref{eq:J3} multiplied by $\dot u$. This gives $2(M+\kappa/R_\Sigma^{3}-m)\dot R_\Sigma=0$. For $\dot R_\Sigma\ne0$ this fixes $m$, and by continuity it holds also at isolated points where $\dot R_\Sigma=0$. Then $q$ follows from \eqref{eq:A}. Because $\dot u>0$ the map $\tau\mapsto u$ is invertible, so $m$ and $q$ are well-defined functions of $u$. The formula for $\Theta$ follows by inserting $m_u=-3\kappa\,(\dd R_\Sigma/\dd u)/R_\Sigma^4$ and $q_u=-6\kappa\,(\dd R_\Sigma/\dd u)/R_\Sigma^3$. The identity $\kappa/R_\Sigma^3=M\rho/\lambda$ gives the second form of $m$. The normal-normal component is checked directly: it equals $-q/R^4+8\pi\Theta(n^u)^2$ evaluated on $\Sigma$, where $\Theta=0$, and $q/R_\Sigma^4=8\pi\rho^2/\rhoc$ by \eqref{eq:solution}. (This equality also follows from the Gauss equation, which relates $G_{\mu\nu}n^\mu n^\nu$ to $K$ and the intrinsic curvature of $\Sigma$.)
\end{proof}

A symbolic verification of all three junction conditions with \eqref{eq:solution} is summarised in Appendix~\ref{app:B}. The same functions solve the junction problem for the advanced form $ds^2=-f\,\dd v^2+2\,\dd v\,\dd R+R^2\dd\Omega^2$, with $\dot v=1/(\cos\chi_\Sigma-\dot R_\Sigma)$, which is regular as long as $\dot R_\Sigma<\cos\chi_\Sigma$.

\subsection{Consequences}

Since $\kappa>0$ for $\lambda>0$, the tidal charge $q$ is positive. If the sign of the quadratic term were reversed, as in RS II, $\kappa$ and hence $q$ would be negative, which agrees with the negative tidal charge of RS II black holes \cite{Dadhich2000}. The positive sign is a result of the junction conditions with the SS correction. It does not follow from a separate argument about gravity confinement. The tidal charge is not a free constant of the exterior: it is fixed by the cloud through $q=3\kappa/R_\Sigma^2$ and varies in time.

The exterior also cannot be static. A static exterior of the whole boundary worldtube in the class \eqref{eq:chargedVaidya} would have constant $m$ and $q$, whereas \eqref{eq:solution} gives $m_u=-3\kappa R_\Sigma^{-4}\,\dd R_\Sigma/\dd u\ne0$ whenever the boundary moves. A homogeneous dust cloud on the SS brane therefore cannot have a static exterior within this class, in the same way as on the RS brane \cite{Bruni2001}. In the limit $\lambda\to\infty$, $\kappa\to0$, $q\to0$ and $m\to M$, and one recovers the OS matching to Schwarzschild.

The solution carries an effective mass that differs from $M$. The mass function at large $R$, the Bondi mass, is $m(u)=M(1+2\rho/\rhoc)$. It equals $M$ for a dilute cloud and exceeds $M$ when the density is high, up to $M(1+2\rho_b/\rhoc)<3M$ at the bounce. The Misner--Sharp mass of the interior at the boundary is $M(1-\rho/\rhoc)$. The difference is carried by the effective Weyl sector.

The direction of the effective flux depends on the chart. For the retarded exterior, \eqref{eq:Theta} shows that $\Theta$ has the sign of $\dot R_\Sigma$ for $R>R_\Sigma$. During expansion $\Theta>0$ and the Bondi mass decreases, as for outgoing radiation of positive energy. During contraction $\Theta<0$ and the Bondi mass increases, so the effective flux violates the null energy condition. For the advanced exterior the roles are reversed: the null coefficient is
\begin{equation}
\Theta=\frac{2R\,m_v-q_v}{R^{3}}=-\frac{6\kappa\,\dot R_\Sigma(\cos\chi_\Sigma-\dot R_\Sigma)(R-R_\Sigma)}{R^{3}R_\Sigma^{4}},
\end{equation}
which is non-negative during contraction and negative during expansion. The combination that satisfies the effective null energy condition throughout is advanced for contraction and retarded for expansion, and it is also the combination that remains regular across the trapped phases (\cref{sec:global}). The flux is not a matter field. It is the projected bulk Weyl tensor, and whether it is directed inward from past null infinity or outward is a boundary condition of the bulk problem that the brane equations do not fix.

\subsection{Trapping and exterior horizons}\label{sec:trapping}

The quantity $f_\Sigma=1-\sin^2\chi_\Sigma(1+\dot a^2)$ in \eqref{eq:udot} is $\nabla_\mu R\nabla^\mu R$ on $\Sigma$, so $\Sigma$ is trapped when $f_\Sigma<0$, that is when $\dot a^2>\cot^2\chi_\Sigma$.

\begin{proposition}[Trapping criterion]\label{prop:trap}
For $\Xmax\ge1$ the maximum of $\dot a^2$ over a cycle is $\Xmax-1$, attained at $a=a_m$. Hence $\Sigma$ is untrapped for the whole cycle if and only if $\Xmax\sin^2\chi_\Sigma<1$. Moreover
\begin{equation}
\Xmax\sin^{2}\chi_\Sigma=\Bigl(\frac{9\pi\lambda M^{2}}{4}\Bigr)^{1/3},
\label{eq:trapcrit}
\end{equation}
so a trapped region forms on $\Sigma$ if and only if
\begin{equation}
M>M_c\equiv\frac{2}{3\sqrt{\pi\lambda}} ,
\end{equation}
independently of $\rho_0$ and $\chi_\Sigma$. If $M>M_c$, the equation $g(a)-1=\cot^2\chi_\Sigma$ has two roots $a_-<a_1<a_m<a_2<a_+$, $\Sigma$ is trapped for $a_1<a<a_2$ on each half-cycle, and it is untrapped at the bounce, where $f_\Sigma=\cos^2\chi_\Sigma$.
\end{proposition}
\begin{proof}
$\dot a^2=g-1$ is maximal at $a_m$. The identity \eqref{eq:trapcrit} follows from \eqref{eq:scales} and \eqref{eq:Mdef}, since $\rho_0^{2/3}\sin^2\chi_\Sigma=(3M/4\pi)^{2/3}$. The statement on the roots follows from the shape of $g$.
\end{proof}

The criterion depends only on $M\sqrt\lambda$. In general relativity ($\lambda\to\infty$) $M_c\to0$ and every cloud traps. For the bounce itself, $\dot R_\Sigma=0$ implies $f_\Sigma=\cos^2\chi_\Sigma>0$, so the bounce always occurs on an untrapped portion of $\Sigma$. This is consistent with the general constraints on bouncing collapse that a bounce cannot occur strictly inside a trapped region of the boundary \cite{BenAchour2020}.

For untrapped clouds, $M<M_c$, the quantity $\dot u$ is bounded and positive over the full cycle, and the retarded chart covers the exterior of $\Sigma$ at all times. If in addition
\begin{equation}
\text{(H)}\qquad f(u,R)>0\quad\text{for all }R\ge R_\Sigma(u),
\end{equation}
that is, there is no horizon outside the cloud, then all curvature scalars of the exterior are bounded: the Ricci scalar vanishes, $R_{\mu\nu}R^{\mu\nu}=4q^2/R^8$ and the Kretschmann scalar is given below \eqref{eq:Gvaidya}, with $m$, $q$ bounded and $R\ge R_{\Sigma,\min}>0$. Condition (H) is not implied by $M<M_c$. It must be checked for the cloud considered, as in \cref{sec:num}.

For trapped clouds, $M>M_c$, the retarded chart fails when $\dot R_\Sigma=-\cos\chi_\Sigma$, where $\dot u\to\infty$, and the advanced chart fails when $\dot R_\Sigma=\cos\chi_\Sigma$. Each is regular on the half of the cycle that is relevant to it, and a single exterior covering the whole cycle is constructed in \cref{sec:global}.

\section{Global exterior of a trapped cloud}\label{sec:global}

For $M>M_c$ the boundary becomes trapped during the cycle, and neither the retarded nor the advanced chart of \cref{sec:exterior} covers the whole evolution. In this section we construct an exterior spacetime that does, we establish its regularity and energy conditions, and we describe its causal structure. Throughout, $M>M_c$ and $\Xmax>1$.

\subsection{Time symmetry and the trapped interval}

\begin{lemma}[Time symmetry]\label{lem:sym}
Fix the origin of $\tau$ at the bounce, $a(0)=a_-$. Then $a(-\tau)=a(\tau)$, and $\dot a<0$ for $-T/2<\tau<0$, $\dot a>0$ for $0<\tau<T/2$, where $a(\pm T/2)=a_+$.
\end{lemma}
\begin{proof}
Equation $\ddot a=g'(a)/2$ is invariant under $\tau\to-\tau$, and $a(0)=a_-$, $\dot a(0)=0$ are invariant initial data, so uniqueness gives the symmetry. Monotonicity on each half follows from $\dot a^2=g-1>0$ on $(a_-,a_+)$.
\end{proof}

For $M>M_c$ let $a_1<a_2$ be the two roots of $g(a)-1=\cot^2\chi_\Sigma$ in $(a_-,a_+)$ given by \Cref{prop:trap}, and define
\begin{equation}
\tau_i=\int_{a_-}^{a_i}\frac{\dd a}{\sqrt{g(a)-1}},\qquad i=1,2,\qquad 0<\tau_1<\tau_2<\frac T2 .
\label{eq:tauis}
\end{equation}
By \Cref{prop:trap} and \Cref{lem:sym}, $f_\Sigma>0$ for $|\tau|<\tau_1$ and for $\tau_2<|\tau|\le T/2$, while $f_\Sigma<0$ for $\tau_1<|\tau|<\tau_2$, and $f_\Sigma=0$ at $|\tau|=\tau_1,\tau_2$. At $\tau=\tau_i$ one has $\dot R_\Sigma=\cos\chi_\Sigma$, and at $\tau=-\tau_i$ one has $\dot R_\Sigma=-\cos\chi_\Sigma$.

The advanced chart of \cref{sec:exterior} has $\dot v=1/(\cos\chi_\Sigma-\dot R_\Sigma)$ and is regular while $\dot R_\Sigma<\cos\chi_\Sigma$, hence on $-T/2\le\tau<\tau_1$. The retarded chart has $\dot u=1/(\cos\chi_\Sigma+\dot R_\Sigma)$ and is regular on $-\tau_1<\tau\le T/2$ and, on the preceding contraction, on $-T/2\le\tau<-\tau_2$. In particular the advanced chart is regular throughout the contraction and the bounce, and the retarded chart throughout the bounce and the expansion. These are the two intervals we use. The two charts overlap on the untrapped window $|\tau|<\tau_1$ around the bounce.

\subsection{Position of the boundary relative to the apparent horizons}

On the boundary, the functions of \eqref{eq:solution} can be written in terms of the dimensionless density $x=\rho/\rhoc\in(0,1)$ as
\begin{equation}
m=M(1+2x),\qquad q=6MxR_\Sigma=\frac{8\pi\rho_0^{2}\sin^{4}\chi_\Sigma}{\rhoc\,a^{2}},
\label{eq:mqx}
\end{equation}
where the second form of $q$ follows from $\rho R_\Sigma^3=\rho_0\sin^3\chi_\Sigma$ and $M=(4\pi/3)\rho_0\sin^3\chi_\Sigma$. The quantity $q$ is therefore largest at the bounce. The zeros of $f(u,R)$ at a given boundary event are
\begin{equation}
r_\pm=m\pm\sqrt{m^{2}-q},
\label{eq:rpm}
\end{equation}
when $m^2\ge q$. We call these the apparent horizon radii of the exterior at that event.

\begin{proposition}[Horizon radii and the boundary]\label{prop:order}
(i) The discriminant $m^2-q$ is positive at every instant of the cycle if and only if $M>M_c$. For $M<M_c$ it is negative on a non-empty interval of $\tau$ around the instant at which $\rho=\rhoc/4$.

(ii) Let $M>M_c$. Then
\begin{equation}
R_\Sigma>r_+\ \ (\tau_2<|\tau|\le T/2),\qquad r_-<R_\Sigma<r_+\ \ (\tau_1<|\tau|<\tau_2),\qquad R_\Sigma<r_-\ \ (|\tau|<\tau_1).
\end{equation}
In words, the boundary crosses the outer apparent horizon inward at $\tau=-\tau_2$, crosses the inner apparent horizon inward at $\tau=-\tau_1$, bounces at $R_b<r_-$ inside the inner horizon, and crosses the two horizons outward at $\tau=\tau_1$ and $\tau=\tau_2$.
\end{proposition}
\begin{proof}
(i) Since $R_\Sigma^3=3M/(4\pi\rho)=3M/(4\pi\rhoc\,x)$, \eqref{eq:mqx} gives
\begin{equation}
\frac{m^{2}-q}{M^{2}}=(1+2x)^{2}-6\,\varepsilon\,x^{2/3},\qquad \varepsilon\equiv\Bigl(\frac{3}{4\pi\rhoc M^{2}}\Bigr)^{1/3}=\Bigl(\frac{3}{8\pi\lambda M^{2}}\Bigr)^{1/3}.
\end{equation}
The function $\phi(x)=(1+2x)^2x^{-2/3}$ has $\phi'(x)=0$ only at $x=1/4$, where it attains its minimum $\phi(1/4)=\tfrac94\,4^{2/3}$. Hence $m^2-q>0$ for all $x\in(0,1)$ if and only if $6\varepsilon<\tfrac94\,4^{2/3}$, that is $\varepsilon^3<27/32$, which is $M^2>4/(9\pi\lambda)$, that is $M>M_c$. The density passes through $\rhoc/4$ at $a=a_m\in(a_-,a_+)$. For $M<M_c$ the discriminant is therefore negative at that instant and, by continuity, in a neighbourhood of it.

(ii) By (i) the radii $r_\pm$ are continuous functions of $\tau$ with $r_-<m<r_+$, and $f(u,R)=(R-r_-)(R-r_+)/R^2$. Where $f_\Sigma\ne0$ the difference $R_\Sigma-r_\pm$ does not vanish, so it has a fixed sign on each interval on which $f_\Sigma\ne0$. At the start of the contraction, $\tau=-T/2$, $\dot a=0$ and $a=a_+=\Mzero(1-x_+)$, where $x_+<1/4$ because $a_+>a_m$. Then $R_\Sigma=2M(1-x_+)/\sin^2\chi_\Sigma\ge 2M(1-x_+)$, and
\begin{equation}
R_\Sigma-m\ge M(1-4x_+)>0 .
\end{equation}
Since $f_\Sigma=\cos^2\chi_\Sigma>0$ and $R_\Sigma>m$, we have $R_\Sigma>r_+$. This sign persists until the first zero of $f_\Sigma$ at $\tau=-\tau_2$, where necessarily $R_\Sigma=r_+$. On $(-\tau_2,-\tau_1)$, $f_\Sigma<0$, which means $r_-<R_\Sigma<r_+$, and at $\tau=-\tau_1$ the boundary reaches $r_-$. On $(-\tau_1,\tau_1)$, $f_\Sigma>0$ and $R_\Sigma<r_-$. The remaining intervals follow from \Cref{lem:sym}.
\end{proof}

\subsection{Construction of the exterior}

We write $R_b=a_-\sin\chi_\Sigma$, $m_b=M(1+2\rho_b/\rhoc)$, $q_b=q|_{\tau=0}$, $f_b(R)=1-2m_b/R+q_b/R^2$, and $r_{\pm,b}$ for the roots of $f_b$. By \Cref{prop:order}, $R_b<r_{-,b}<r_{+,b}$. The tortoise coordinate of the static geometry is $R_\ast(R)=\int_{R_b}^{R}\dd R'/f_b(R')$, which is finite and increasing on $[R_b,r_{-,b})$ and diverges as $R\to r_{-,b}$.

\begin{enumerate}
\item[\textbf{A}] (contraction) Advanced Bondi coordinates $(v,R)$ with $\dd v/\dd\tau=1/(\cos\chi_\Sigma-\dot R_\Sigma)$ along $\Sigma$ for $-T/2\le\tau\le0$, normalised by $v(0)=v_b$. The region is $\{v\le v_b,\ R\ge R_\Sigma(v)\}$ with metric
\begin{equation}
ds^{2}=-f_A\,\dd v^{2}+2\,\dd v\,\dd R+R^{2}\dd\Omega^{2},\qquad f_A=1-\frac{2m(v)}{R}+\frac{q(v)}{R^{2}},
\end{equation}
where $m(v)$ and $q(v)$ are \eqref{eq:solution} evaluated at $\tau(v)$.
\item[\textbf{B}] (expansion) The image of A under time reversal. Retarded coordinates $(u,R)$ with $\dd u/\dd\tau=1/(\cos\chi_\Sigma+\dot R_\Sigma)$ for $0\le\tau\le T/2$ and $u(0)=u_b\equiv v_b$. The region is $\{u\ge u_b,\ R\ge R_\Sigma(u)\}$ with metric \eqref{eq:Bondi}, \eqref{eq:chargedVaidya} and $m(u)$, $q(u)$ from \eqref{eq:solution}.
\item[\textbf{W}] (static wedge) The part of the static charged geometry $f_b$ in advanced coordinates,
\begin{equation}
W=\bigl\{(v,R):\ v\ge v_b,\ R\ge R_b,\ \text{and}\ v-v_b\le 2R_\ast(R)\ \text{if}\ R<r_{-,b}\bigr\}.
\end{equation}
\end{enumerate}
The pieces are joined as follows. A and W share the coordinates $(v,R)$ and are glued along the null hypersurface $N^-=\{v=v_b\}$. W and B are glued along the null hypersurface $N^+=\{u=u_b,\ R\ge R_b\}$ by the identification $v=u+2R_\ast(R)$ for $R<r_{-,b}$. The three pieces meet at the bounce event $\mathrm B=(v_b,R_b)$ on $\Sigma$. The exterior is
\begin{equation}
\mathcal E=A\cup W\cup B .
\end{equation}
It has two null boundary pieces that are not glued to anything: the part of the static geometry at $v\to\infty$, which is the Cauchy horizon of the wedge, and the part of $N^+$ with $R>r_{-,b}$ on the side of $W$. Three features of this construction are choices made within the Bondi class and are not implied by the brane equations. The effective flux is directed inward during the contraction and outward during the expansion, which requires a corresponding boundary condition on the bulk Weyl field at past null infinity. The wedge $W$ is the static charged geometry, which is the minimal continuation compatible with $C^1$ gluing at $N^\pm$. Finally, $\mathcal E$ is not extended across the two null boundary pieces. The region beyond the Cauchy horizon of the wedge is not determined by the data on $\Sigma$, and the inner horizon of a static charged geometry is known to be unstable under perturbations, so whether the dynamical exterior is stable is a separate question that we do not address. The existence of a regular bulk geometry that induces $\mathcal E_{\mu\nu}$ is likewise not shown.

\begin{proposition}[Global exterior]\label{prop:global}
Let $M>M_c$. The spacetime $\mathcal E$ has the following properties.
\begin{enumerate}
\item[(a)] The metric is smooth on $A$, $W$, $B$ and of class $C^{1,1}$ across $N^-$ and $N^+$. The curvature tensor is continuous across them. Hence the effective brane equations \eqref{eq:vacuum} hold on $\mathcal E$ with $\mathcal E_{\mu\nu}=-G_{\mu\nu}$ trace-free and without surface layers.
\item[(b)] $\Sigma$ is a timelike boundary of $\mathcal E$ for the whole interval $-T/2\le\tau\le T/2$, and the Israel conditions of \cref{sec:exterior} hold on all of it, including the trapped intervals $\tau_1<|\tau|<\tau_2$.
\item[(c)] $R\ge R_b>0$ on $\mathcal E$. All curvature scalars are bounded: $\Ric=0$, $R_{\mu\nu}R^{\mu\nu}=4q^2/R^8$, and $R_{\alpha\beta\gamma\delta}R^{\alpha\beta\gamma\delta}=48m^2/R^6-96mq/R^7+56q^2/R^8$, with $m\le m_b$ and $q\le q_b$.
\item[(d)] With $T^{\rm eff}_{\mu\nu}\equiv-\mathcal E_{\mu\nu}/8\pi=G_{\mu\nu}/8\pi$, the effective null, weak and dominant energy conditions hold on $\mathcal E$.
\end{enumerate}
\end{proposition}
\begin{proof}
(a) On $A$, $m_v=-3\kappa\,(\dd R_\Sigma/\dd v)/R_\Sigma^{4}$ and $q_v=2R_\Sigma m_v$, with $\dd R_\Sigma/\dd v=\dot R_\Sigma(\cos\chi_\Sigma-\dot R_\Sigma)$. At the bounce $\dot R_\Sigma=0$, so $m_v=q_v=0$ at $v=v_b$, and $m(v_b)=m_b$, $q(v_b)=q_b$. Extending $m$ and $q$ by the constants $m_b$, $q_b$ for $v>v_b$ therefore gives $C^1$ functions, whose second derivatives have a finite jump. The metric components $f(v,R)$ are then $C^{1,1}$ in $v$. Curvature components of the Bondi form contain $m$, $q$, $m_v$, $q_v$ and $R$ but no second derivatives of $m$ and $q$ (Appendix~\ref{app:A}), so they are continuous. The same argument applies at $N^+$ in retarded coordinates, using $m_u=q_u=0$ at $u_b$. Since $G_{\mu\nu}$ is continuous, there is no distributional source on $N^\pm$. Trace-freeness and conservation of $\mathcal E_{\mu\nu}$ hold piecewise because $\Ric=0$ for the charged Bondi form, and they extend across $N^\pm$ by the Bianchi identity applied to the $C^{1,1}$ metric.

(b) On $A$ the advanced chart is regular for $\tau\le0$ by the discussion following \eqref{eq:tauis}, and the junction identities were verified for it in Appendix~\ref{app:B}. On $B$ the retarded chart is regular for $\tau\ge0$, and the identities hold there by the same verification, or by time reversal.

(c) On $A$, $R\ge R_\Sigma(v)\ge R_b$ because $R_\Sigma=a\sin\chi_\Sigma$ has its minimum at the bounce. The same holds on $B$, and on $W$ by definition. By \eqref{eq:mqx}, $m$ increases with $\rho$ and $q\propto a^{-2}$, so both are largest at the bounce. The scalar expressions follow from \eqref{eq:chargedVaidya} and Appendix~\ref{app:A}.

(d) The Einstein tensor of the Bondi form is the sum of a Maxwell-type term and a null term, \eqref{eq:Gvaidya}. The Maxwell-type term with $q>0$ has energy density $q/(8\pi R^4)>0$ for every observer and satisfies the dominant energy condition. The null term is $(\Theta/8\pi)k_\mu k_\nu$ and satisfies all three conditions if $\Theta\ge0$. On $A$ the null coefficient is
\begin{equation}
\Theta_A=\frac{2R\,m_v-q_v}{R^{3}}=-\frac{6\kappa\,\dot R_\Sigma\,(\cos\chi_\Sigma-\dot R_\Sigma)\,(R-R_\Sigma)}{R^{3}R_\Sigma^{4}}\ \ge0,
\end{equation}
because $\dot R_\Sigma\le0$ and $R\ge R_\Sigma$ and $\cos\chi_\Sigma-\dot R_\Sigma>0$. On $B$ it is $\Theta_B$ of \eqref{eq:Theta}, which is non-negative because $\dot R_\Sigma\ge0$ and $\cos\chi_\Sigma+\dot R_\Sigma>0$. On $W$ the null term vanishes. A sum of two tensors that each satisfy the dominant energy condition satisfies it, and the dominant condition implies the weak and null conditions.
\end{proof}

\begin{remark}
In the retarded chart alone the contraction would carry $\Theta<0$ (\cref{sec:exterior}). The choice made here, advanced for contraction and retarded for expansion, is therefore the one for which the effective source satisfies the energy conditions, and it is also the only choice within the Bondi class for which the chart stays regular across the trapped phase. Whether the effective flux is directed inward from past null infinity during contraction is a boundary condition on the bulk problem, which the brane equations do not fix.
\end{remark}

\subsection{Causal structure}

\begin{proposition}[Black hole region and re-emergence]\label{prop:causal}
Let $M>M_c$.
\begin{enumerate}
\item[(a)] For $v>v_b$ the region $\{R>r_{+,b}\}$ of $W$ is a static charged exterior with mass parameter $m_b=M(1+2\rho_b/\rhoc)$ and charge parameter $q_b$, with $m_b^2>q_b$. No future-directed causal curve in $\mathcal E$ starting at a point of $W$ with $r_{-,b}<R<r_{+,b}$ reaches the late-time region $\{R>r_{+,b}\}$ of $W$. In particular $\mathcal E$ contains a black hole region for observers in that asymptotic region.
\item[(b)] The boundary $\Sigma$ crosses the outer apparent horizon outward at $\tau=\tau_2$ and satisfies $R_\Sigma>r_+$ for $\tau_2<\tau\le T/2$. For $u\ge u_b$ every outgoing null ray $u=\text{const}$ leaving $\Sigma$ reaches $R\to\infty$ in $B$. Observers in the asymptotic region of $B$ therefore see the cloud re-expand through a white hole region, and receive the effective flux $\Theta_B\ge0$.
\item[(c)] The effective energy absorbed from past null infinity during the contraction and the effective energy emitted to future null infinity during the expansion are equal,
\begin{equation}
\int m_v\,\dd v=\int(-m_u)\,\dd u=m_b-M(1+2x_+)=2M\,(x_b-x_+),\qquad x_b=\frac{\rho_b}{\rhoc},\ x_+=\frac{\rho_+}{\rhoc}.
\end{equation}
\end{enumerate}
\end{proposition}
\begin{proof}
(a) The part of $W$ with $R>r_{+,b}$ is, by definition, a piece of the static charged geometry with parameters $m_b$, $q_b$, and $m_b^2>q_b$ by \Cref{prop:order}(i) at $\tau=0$. Let $p\in W$ with $r_{-,b}<R<r_{+,b}$. There $f_b<0$, so $R$ strictly decreases along every future-directed causal curve. Such a curve therefore cannot reach $R>r_{+,b}$ while it stays in $W$. It can leave $W$ only through $N^{+}$, which is a hypersurface $u=\text{const}$ that future-directed causal curves cross only towards larger $u$, that is into $B$, or it can approach the Cauchy horizon at $v\to\infty$. Since $N^-$ is crossed only from $A$ into $W$ and $N^+$ only from $W$ into $B$, a curve that has left $W$ cannot return to it.

(b) By \Cref{prop:order}(ii) and the time symmetry, $R_\Sigma>r_+$ for $\tau>\tau_2$, and $R_\Sigma=r_+$ at $\tau_2$ with $\dot R_\Sigma>0$. The null rays $u=\text{const}$ are the curves $\dd u=0$ of the metric \eqref{eq:Bondi}. They extend to $R\to\infty$ because the retarded chart is regular at all $R\ge R_\Sigma(u)$. The sign of $\Theta_B$ was established in \Cref{prop:global}(d).

(c) By \eqref{eq:solution}, $m(v)=M(1+2x)$ increases from $M(1+2x_+)$ at $\tau=-T/2$ to $m_b$ at the bounce, and $m(u)$ decreases back to $M(1+2x_+)$ at $\tau=T/2$ by time symmetry. The flux through a large sphere at fixed advanced time is $4\pi R^2\,\Theta_A/8\pi\to m_v$, and similarly for the retarded flux.
\end{proof}

Figure~\ref{fig:five} shows the causal structure in a qualitative diagram, drawn on the static charged geometry that $W$ is cut from. The positions of the apparent horizons in $A$ and $B$ differ from the straight lines of the static diagram because $m$ and $q$ vary there, but their ordering relative to $\Sigma$ is that established in \Cref{prop:order}.

\begin{figure}[t]
\centering
\includegraphics[height=0.42\textheight]{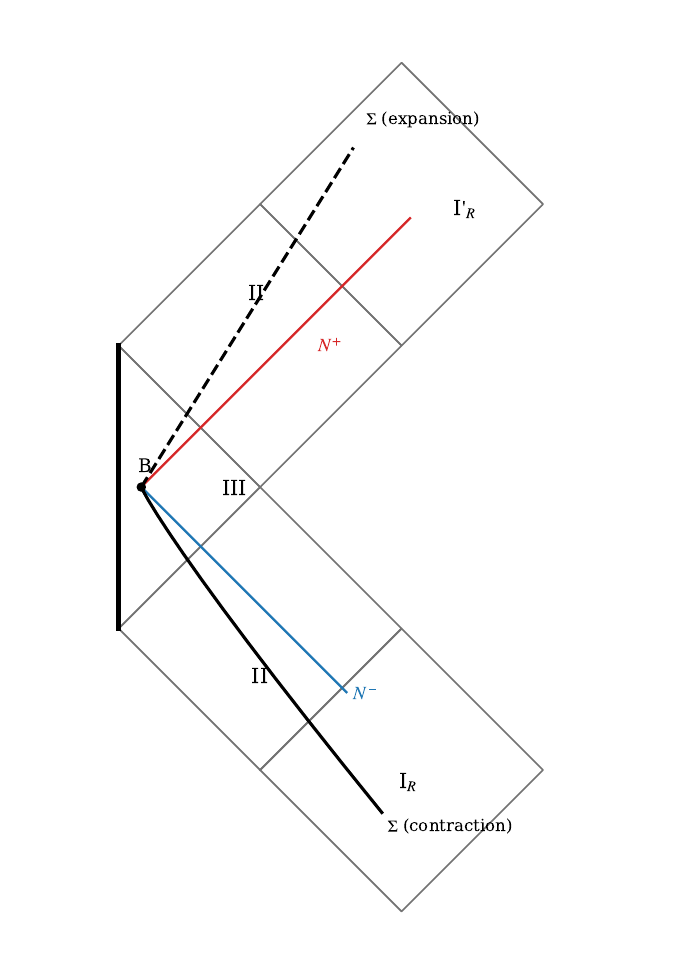}
\caption{Qualitative causal diagram of $\mathcal E$ for a trapped cloud. The thick vertical line is the centre of the cloud, whose boundary $\Sigma$ is drawn solid during contraction and dashed during expansion. I$_R$ is the asymptotic region seen at the start, II is the trapped region between $r_-$ and $r_+$, III is the inner region $R<r_-$ in which the bounce B occurs, II$'$ is the white hole region and I$'_R$ the second asymptotic region. $N^-$ and $N^+$ are the null hypersurfaces $v=v_b$ and $u=u_b$ through B.}
\label{fig:five}
\end{figure}

\subsection{Numerical example}

For the trapped cloud of \cref{sec:num}, $\rho_0=\lambda$ and $\chi_\Sigma=0.8$ (units $\lambda=1$), the roots of $g-1=\cot^2\chi_\Sigma$ are $a_1=0.8543$ and $a_2=4.284$, with $\tau_1=0.0650$ and $\tau_2=2.504$. The half period is $T/2=13.19$. At $\tau=-\tau_2$ the boundary has $R_\Sigma=3.073=r_+$, and at $\tau=-\tau_1$ it has $R_\Sigma=0.6135\simeq r_-=0.6129$. At the bounce $R_b=0.5893$, $m_b=4.336=2.804\,M$, $q_b=4.931$, and the static radii are $r_{-,b}=0.6119$ and $r_{+,b}=8.059$, so that $R_b<r_{-,b}$ as required. The minimum of $m^2-q$ over the cycle is $2.171>0$. The effective energy exchanged with the two asymptotic regions is $m_b-M=2.790$ in units of $\lambda^{-1/2}$. Figure~\ref{fig:four} shows $R_\Sigma$, $r_\pm$, $m$ and $q$ along the cycle.

\begin{figure}[t]
\centering
\includegraphics[width=0.96\textwidth]{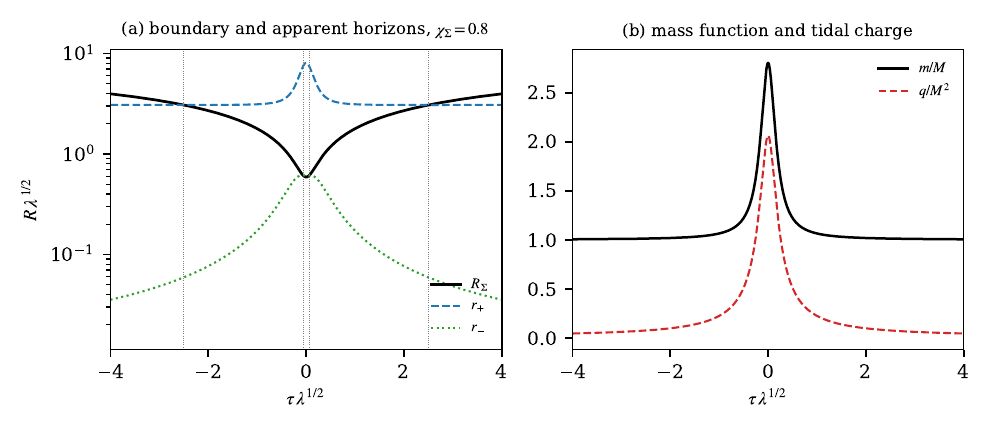}
\caption{Trapped cloud with $\rho_0=\lambda$, $\chi_\Sigma=0.8$. (a) Boundary radius $R_\Sigma$ (solid) and apparent horizon radii $r_+$ (dashed), $r_-$ (dotted). Vertical lines mark $|\tau|=\tau_1,\tau_2$. The boundary is between the two horizons for $\tau_1<|\tau|<\tau_2$ and inside $r_-$ at the bounce. (b) Mass function $m/M$ and tidal charge $q/M^2$ on $\Sigma$.}
\label{fig:four}
\end{figure}

\section{Numerical illustration}\label{sec:num}

We work in units $\lambda=1$, so that lengths are in units of $\lambda^{-1/2}$. Take $\rho_0=\lambda$, for which $\astar=0.4886$ and $\Xmax=4.987$. The roots are $a_-=0.8215$ and $a_+=8.370$, the bounce density is $\rho_b/\rhoc=0.902$, and the period is $T=26.37$. At the bounce $\Ric_b=127.1\,\lambda$, in agreement with $8\pi\rho_b(1+2\rho_b/\rhoc)$. Along the cycle the maxima of the invariants are $5.04\times10^{3}\lambda^2$ for $R_{\mu\nu}R^{\mu\nu}$ and $4.69\times10^{3}\lambda^2$ for the Kretschmann scalar, below the bounds of \Cref{prop:curv}. The critical mass is $M_c=0.376$.

Table~\ref{tab:clouds} gives two clouds with the same $\rho_0$. For $\chi_\Sigma=0.3$ the mass is $M=0.108<M_c$ and the boundary is never trapped. The exterior horizons $r_\pm(u)=m\pm\sqrt{m^2-q}$ either do not exist or lie inside $R_\Sigma$ at all times, so (H) holds (\cref{fig:three}a). For $\chi_\Sigma=0.8$ the mass is $M=1.546>M_c$. The boundary is trapped for $0.065<|\tau|<2.50$ and untrapped in the short window $|\tau|<0.065$ around the bounce (\cref{fig:three}b).

\begin{table}[t]
\centering
\begin{tabular}{lcc}
\toprule
 & $\chi_\Sigma=0.3$ & $\chi_\Sigma=0.8$ \\
\midrule
Mass $M$ & $0.1081$ & $1.546$ \\
$M/M_c$ & $0.29$ & $4.11$ \\
$\min_\tau R_\Sigma$ & $0.2428$ & $0.5893$ \\
$\min_\tau f_\Sigma$ & $0.564$ & $-1.566$ \\
Trapped on $\Sigma$ & no & yes \\
Exterior condition (H) & satisfied & not applicable \\
$\max_u m/M=1+2\rho_b/\rhoc$ & $2.804$ & $2.804$ \\
\bottomrule
\end{tabular}
\caption{Two clouds with $\rho_0=\lambda=1$ (units $\lambda=1$).}
\label{tab:clouds}
\end{table}

\begin{figure}[t]
\centering
\includegraphics[width=0.96\textwidth]{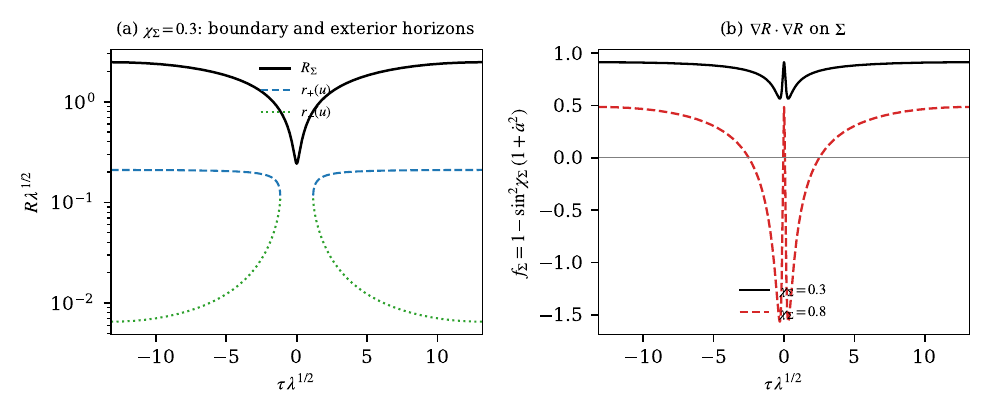}
\caption{(a) Boundary radius $R_\Sigma$ (solid) for $\chi_\Sigma=0.3$ and exterior horizon radii $r_+$ (dashed) and $r_-$ (dotted) where they exist. The boundary stays outside $r_+$ whenever horizons exist. (b) The quantity $f_\Sigma=1-\sin^2\chi_\Sigma(1+\dot a^2)$ for $\chi_\Sigma=0.3$ (solid) and $\chi_\Sigma=0.8$ (dashed). Negative values mean that $\Sigma$ is trapped.}
\label{fig:three}
\end{figure}

A scan of $2400$ clouds with $\rho_0$ between $\rho_{0,\min}$ and $100\lambda$ and $\chi_\Sigma$ between $0.02$ and $1.55$ confirmed \Cref{prop:trap} in every case. Condition (H) failed for some compact untrapped clouds ($M<M_c$ with $\chi_\Sigma\gtrsim1$), in which the boundary lies inside an exterior inner horizon. These configurations lie outside the validity of the retarded-chart statements.

\section{Discussion}\label{sec:disc}

The main conclusion is that a homogeneous dust cloud on the SS brane does not reach a curvature singularity in the interior. The mechanism is the negative quadratic term in the Friedmann equation, which corresponds to an effective negative pressure growing with density. It is purely classical. It reverses the sign of the deceleration at $\rho=\rhoc/4$ and turns the closed-universe contraction into a bounce at a radius that is bounded below by $\astar=\sqrt{3/4\pi\lambda}$. The effective strong energy condition fails once $\rho>\rhoc/4$, which is what allows the bounce, and the effective null energy condition fails for $\rho>\lambda$. The curvature bound of \Cref{prop:curv} plays a role analogous to the Planckian bound of loop-quantum-gravity bounces \cite{BenAchourUzan2020}, but it is set by the brane tension and arises from higher-dimensional geometry.

Two points differ from what one might guess from the RS case. First, the universal radius $\astar$ and the density $\lambda$ are properties of the saturating configuration only. A given cloud bounces at a radius $a_-\ge\astar$ that grows with $\rho_0$, and at a density below $\rhoc$. Second, the tidal charge in the exterior is positive and time dependent. It is fixed by the junction conditions, and it is $q=3\kappa/R_\Sigma^2$ rather than a constant parameter of a static solution.

For trapped clouds the exterior constructed in \cref{sec:global} describes a black hole for distant observers in the first asymptotic region, with a trapped region between two apparent horizons, a bounce inside the inner horizon and a re-expansion through a white hole region into a second asymptotic region. The curvature scalars are bounded everywhere on this exterior, and the effective source satisfies the dominant energy condition. The effective mass seen from infinity rises to $M(1+2\rho_b/\rhoc)$ at the bounce, which is up to three times the cloud mass, and the excess is returned to the second asymptotic region. This is the brane analogue of the black-to-white hole bounce studied for quantum gravity corrections \cite{BenAchour2020}, with the difference that the repulsive mechanism is classical and the exterior is dynamical.

The analysis refers to an exactly homogeneous and spherical cloud, and completeness of the interior has been shown for the radial comoving congruence only. Stability of the bounce against inhomogeneous or non-spherical perturbations, and the Kaluza--Klein spectrum for a timelike extra dimension (\cref{sec:timelike}), are not covered. The replacement of the homogeneous interior by a Lema\^itre--Tolman--Bondi solution, which allows one to ask whether the repulsive correction also protects against shell crossing and shell focusing, is the subject of a follow-up project that we expect to report in the near future.

\section*{Acknowledgments}
The authors thank Prof.  Patrick Peter (Institut d'Astrophysique de Paris, France) and Prof. Varun Sahni (Inter University Centre for Astronomy and Astrophysics, India) for useful discussions and suggestions.

\appendix

\section{Einstein tensor of the Bondi form}\label{app:A}

For the metric $ds^2=-f\,\dd u^2-2\,\dd u\,\dd R+R^2\dd\Omega^2$ with $f=1-2m(u,R)/R$,
\begin{align}
G_{uu}&=\frac{2}{R^{3}}\bigl(R\,m_R-R\,m_u-2m\,m_R\bigr),& G_{uR}&=\frac{2m_R}{R^{2}},\\
G_{RR}&=0,& G_{\theta\theta}&=-R\,m_{RR},\\
\Ric&=\frac{2}{R^{3}}\,\partial_R\!\bigl(R^{2}m_R\bigr).&&
\end{align}
For $m=m(u)-q(u)/(2R)$ (that is, \eqref{eq:chargedVaidya}) these reduce to
\begin{equation}
G_{uu}=\frac{q\,f}{R^{4}}+\frac{q_u-2R\,m_u}{R^{3}},\quad G_{uR}=\frac{q}{R^{4}},\quad G_{RR}=0,\quad G_{\theta\theta}=\frac{q}{R^{2}},
\end{equation}
which is \eqref{eq:Gvaidya}. In both charts the Christoffel symbols contain only first derivatives of $f$ with respect to $u$ or $v$, so the Riemann tensor contains $m_u$ and $q_u$ (respectively $m_v$, $q_v$) but not their second derivatives. For the advanced form $ds^2=-f\,\dd v^2+2\,\dd v\,\dd R+R^2\dd\Omega^2$ the same expressions hold with $m_u\to-m_v$ in $G_{vv}$ and $G_{vR}\to-G_{vR}$.

\section{Verification of the junction conditions}\label{app:B}

With $n_\mu=(-\dot R,\dot u)$, $U^\mu=(\dot u,\dot R)$ and the charged Vaidya metric, the exterior quantities are
\begin{align}
G_{\mu\nu}n^\mu U^\nu&=\frac{\dot u^{2}}{R^{3}}\bigl(2R\,m_u-q_u\bigr),\\
G_{\mu\nu}n^\mu n^\nu&=-\frac{\dot u}{R^{6}}\Bigl[\dot u\bigl(2R^{4}m_u-R^{3}q_u+R^{2}q-2Rmq+q^{2}\bigr)+2R^{2}\dot R\,q\Bigr],
\end{align}
and $K_{\tau\tau}=-n_\mu\bigl(\ddot x^\mu+\Gamma^\mu{}_{\alpha\beta}U^\alpha U^\beta\bigr)$. Substitute \eqref{eq:solution}, $\dot u=1/(\cos\chi_\Sigma+\dot R)$, $\ddot u=-\ddot R/(\cos\chi_\Sigma+\dot R)^2$, $m_u=\dot m/\dot u$, $q_u=\dot q/\dot u$ with $\dot m=-3\kappa\dot R/R^4$ and $\dot q=-6\kappa\dot R/R^3$, and use the interior relations $\dot R^2=\sin^2\chi_\Sigma\,[\,a^2(8\pi/3)\rho(1-\rho/\rhoc)-1\,]$ and $\ddot R=\sin\chi_\Sigma\,\ddot a$ with \eqref{eq:Ray}. Then
\begin{equation}
K_{\tau\tau}=0,\qquad G_{\mu\nu}n^\mu U^\nu=0,\qquad G_{\mu\nu}n^\mu n^\nu=-\frac{8\pi\rho^{2}}{\rhoc},
\end{equation}
identically in $(\rho_0,\rhoc,\chi_\Sigma,R_\Sigma)$. We checked the three identities symbolically and numerically for the retarded chart with both signs of $\dot R$, and for the advanced chart with $\dot v=1/(\cos\chi_\Sigma-\dot R)$ and $\dot R<0$. The expanding branch of the advanced chart is not needed.

\end{document}